\documentclass[journal]{IEEEtran}
\usepackage{amsthm}
\usepackage{amsmath}
\usepackage{amssymb}
\usepackage{latexsym}
\usepackage{bbm}
\usepackage[dvips]{graphicx}
\usepackage{enumerate}
\usepackage{verbatim}
\usepackage{amsfonts}
\usepackage{epsfig}
\usepackage{tikz}
\usetikzlibrary{fit}
\usetikzlibrary{shapes}
\usepackage{url}
\usepackage[lofdepth,lotdepth,caption=false]{subfig}

\usepackage{color}
\usepackage{epstopdf}
\usepackage{breqn}

\newtheorem{lemma}{Lemma}

\begin{document}
%
\title{On Quantizer Design to Exploit Common Information in Layered Coding of Vector Sources}
%
%
%

\author{Mehdi~Salehifar,~\IEEEmembership{Member,~IEEE,}
        Tejaswi~Nanjundaswamy,~\IEEEmembership{Member,~IEEE,}
        and~Kenneth~Rose,~\IEEEmembership{Fellow,~IEEE}
}

%
%

\markboth{}%
{}
%



\maketitle

\begin{abstract}
This paper  studies a layered coding framework with a relaxed hierarchical structure. Advances in wired/wireless communication and consumer electronic devices have created a requirement for serving the same content at different quality levels. The key challenge is to optimally encode all the required quality levels with efficient usage of storage and networking resources. The approach  to store and transmit independent copies for every required quality level is highly wasteful in resources. Alternatively, conventional scalable coding has inherent loss due to its structure. This paper studies a layered coding framework with a relaxed hierarchical structure to transmit information common to different quality levels along with individual bit streams for each quality level. The flexibility of sharing only a properly selected subset of information from a lower quality level with the higher quality level, enables achieving operating points between conventional scalable coding and independent coding, to control the layered coding penalty. Jointly designing common and individual layers' coders overcomes the limitations of conventional scalable coding and non-scalable coding, by providing the flexibility of transmitting common and individual bit-streams for different quality levels. It extracts the common information between different quality levels with negligible performance penalty. Simulation results for  practically important sources, confirm the superiority of the work.

\end{abstract}


%
\IEEEpeerreviewmaketitle

\section{Introduction}
\IEEEPARstart{T}{echnological}  advances ranging from multigigabit high-speed Internet to wireless communication and mobile, limited resource receivers, have created an extremely heterogeneous network scenario with data consumption devices of highly diverse decoding and display capabilities, all accessing the same content over networks of time varying bandwidth and latency.
The primary challenge is to maintain optimal signal quality for a wide variety of users, while ensuring efficient use of resources for storage and transmission across the network. 

The simplest solution to address this challenge is storing and transmitting independent copies of the signal for every type of user the provider serves. This solution is highly wasteful in resources and results in extremely poor scalability. 

In an alternative solution, conventional scalable coding \cite{svc,svc2} generates layered bit-streams, wherein a base layer provides a coarse quality reconstruction and successive layers refine the quality, incrementally. Depending on the network, channel and user constraints, a suitable number of layers is transmitted and decoded, yielding a prescribed quality level. 
However, it is widely recognized that there is an inherent loss due to the scalable coding structure, with significantly worse distortion compared to independent (non-scalable) encoding at given receive rates \cite{Fail_SR,Fail_SR2,CELQ},
as most sources are not successively refinable at finite delays for the distortion measure employed and the combination of rates at each layer.
Moreover for fixed receive rates, non-scalable coding and conventional scalable coding have the highest and the lowest total transmit rate, respectively.
Thus, non-scalable coding and conventional scalable coding represent two extreme points in the trade off between total transmit rate and distortions at the decoders, with fixed receive rates.

In our previous work we proposed a novel layered coding paradigm for multiple quality levels \cite{CI_NSR} inspired by the information theoretic concept of common information of dependent random variables \cite{G-W,wyner,GW_Common}, wherein only a (properly selected) subset of the information at a lower quality level is shared with the higher quality level. This flexibility enables efficiently extracting common information between quality levels and achieve intermediate operating points in the trade off between total transmit rate and distortions at the decoders, in effect controlling the layered coding penalty. Our early results \cite{CI_NSR} established the information theoretic foundations for this framework and a later paper \cite{layered} employed this framework within a standard audio coder to demonstrate its potential. In this paper we tackle the important problem of designing quantizers for this layered coding framework.

 First we design quantizers for two quality levels with fixed receive rates. We need to design three quantizers, one for the common layer, whose output is sent to both the decoders, and two other quantizers refining the common layer information at two quality levels, whose output is sent individually to the two decoders. We first propose a technique to jointly design quantizers across layers of this framework, while focusing on the setting of two quality levels.
 We propose an iterative approach for designing the three quantizers, wherein at each iteration one quantizer is updated to minimize the overall cost function while the others are fixed, and the iterations are repeated until convergence. Given a common layer quantizer, for each of the individual layer, refining quantizers for all the intervals of the common layer are jointly designed to minimize the overall cost function. We also develop (``Lloyd algorithm style'') optimal update rules for the common layer quantizer that minimizes the overall cost while accounting for the current individual layer quantizers.
 
  Then we propose a low complexity approach for deriving the quantizers.  First we employ an optimal quantizer for a given rate at the common layer. Given this quantizer, we design two other optimal quantizers at two different required rates, conditioned on each common layer interval. Finally the optimal common layer rate is estimated numerically by trying multiple allowed common rates and selecting the highest one amongst those with negligible loss in distortion compared to non-scalable coding. We then adapt this technique to the practically important Laplacian sources.
 
 In our previous works \cite{icasspcommon,phd,dcccom} we explained the low complexity design of quantizers as well as briefly explained the joint design of quantizers for two quality levels.

Finally we propose an iterative technique for joint design of vector quantizers for all layers of this framework. For simplicity, we first explain the approach for the setting of two quality levels. Then we explain  the extended approach  to other relaxed hierarchical structures. We  develop a cost function which explicitly controls the tradeoff between distortions, receive rates, and total transmit rate.
We then propose an iterative approach for jointly designing vector quantizers for all the layers, wherein we estimate optimal quantizer partitions at all the layers, given reconstruction codebooks, and optimal reconstruction codebooks for all quality levels, given quantizer partitions, iteratively, until convergence. 

Experimental evaluation results for multivariate normal distribution and Laplacian distribution, substantiate the usefulness of the proposed technique.

The rest of this paper is organized as follows. In Section II, we present preliminaries. In Sections III and IV we describe the proposed methods. Experimental results are presented in Section V and concluding remarks are in Section VI.

\begin{section}{Preliminaries }
	
Few mathematical results have had as much impact on the foundation of the information age as Shannon's 1948 point-to-point communication theorems \cite{pro6}. However, the communication model assumed in these seminal contributions is inadequate for the realities of modern networks. Extensions of the theory to multi-terminal settings have proven difficult and, despite several spectacular advances, many questions remain only partially answered, and many more questions regarding the conversion of available insights into practical approaches, remain unanswered.
		\begin{subsection}{Successive Refinement and Scalable Coding}
			Rate-distortion theory is a major branch of information theory, which provides the theoretical foundations for lossy data compression. 
			
%
%
%
%
%
%
			 
			  The fundamental theorem of rate-distortion theory \cite{pro6} is that the minimum rate of communication required to convey the sequence of iid random variables, so that the receiver can reconstruct the sequence at average distortion of at most $D$, is given by:
			  
			  \begin{align}
			  R(D)= \min_{p(\hat{x}|x): E\{d(x,\hat{x})\}\leq D}I(X;\hat{X})
			  \end{align} 
			  
			  The result states that the rate-distortion function, indicating the minimum achievable rate for prescribed distortion, is given by minimizing the mutual information over all conditional distributions  subject to the distortion constraint. It led to an abundance of research work devoted to finding
			  rate-distortion bounds for various new settings \cite{pro8,G-W,pro10,pro11,pro12}, numerical evaluation of the rate  distortion function for generic sources and distortion measures \cite{pro13,pro14,pro15,pro16} and on practical scalar/vector quantizer design and analysis \cite{pro17,pro18,pro19}.
			  
			  From the view point of rate-distortion theory, scalable coding has been addressed in the context of successive refinement of information \cite{pro21,pro22,pro23,pro24,pro25,pro26,pro27}. 
			  The problem is motivated by scalable coding, where the encoder generates two layers of information, namely, the base layer at rate $R_{12}$, and the enhancement layer at rate $R_2$ (where a rate subscript specifies to which decoders it is routed). The base layer provides a coarse reconstruction of the source (at rate $R_{12}$), while the enhancement layer is used to `refine' the reconstruction beyond the base layer (at an overall rate of $R_2 + R_{12}$). The base and enhancement layer distortions are $D_1$ and $D_2$ respectively, where $D_2 < D_1$.
			  
			  Equitz and Cover \cite{pro21} established the conditions  for a source-distortion pair to achieve rate-distortion optimality at both the layers simultaneously. Such source-distortion pairs are called successively refinable in the literature.
			  
			  This optimality can be achieved for a given distortion measure at distortions $D_1$
			  and $D_2$ if and only if there exists a conditional probability distribution $p(\hat{x}_1,\hat{x}_2|x)$ such that: 
			  \[
			  E\{d(X,\hat{X}_1)\}\leq D_1 \qquad  E\{d(X,\hat{X}_2)\}\leq D_2
			  \]
			  \begin{align}
			  I(X;\hat{X}_1)=R(D_1) \qquad  I(X;\hat{X}_2)=R(D_2)
			  \end{align}
			  \[
			  X \leftrightarrow \hat{X}_2 \leftrightarrow \hat{X}_1
			  \]
			  
			  However, when a source-distortion pair is not successively refinable, it is impossible to maintain  optimality at both the layers simultaneously in the scalable coding framework.
			
		\end{subsection}

\subsection{Gray-Wyner Network}
		\label{subsec:GW}
		The GW network, consists of an encoder that transmits two correlated sources to two receivers using three channels: a common channel linking the encoder to both receivers, and two private channels linking it to individual receivers. The channels are assumed to be noiseless and each channel has a specified per bit communication cost (Fig.~\ref{fig:Gray-Wyner}).  
		
		\begin{figure}[t]
		\centering
		\centerline{\includegraphics[width=1\linewidth]{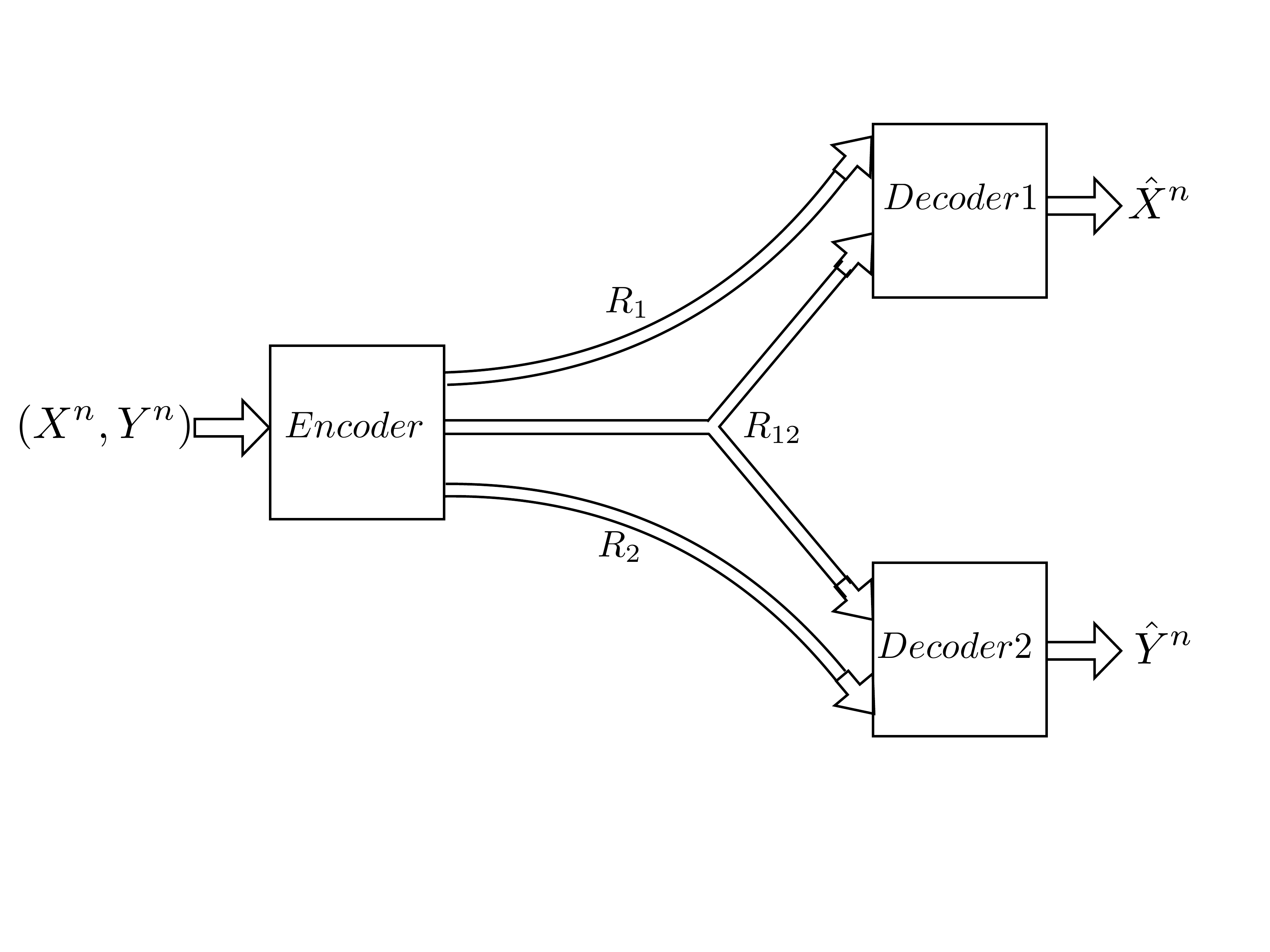}}
			\vspace{-1cm}
		\caption{Gray-Wyner network }
		\label{fig:Gray-Wyner}
	\end{figure}

		The objective is to minimize the communication cost while achieving decoder reconstructions below the allowed distortion. This problem clearly involves rate trade offs as we seek the optimal rate triple ($R_{12},R_1,R_2$). GW derived the asymptotic minimum cost achievable for this network. 

		{\bf{Definitions of Common Information based on the GW Network}} \cite{wyner}: 
		
		A different concept of common information (CI) of two dependent random variables was proposed by G$\mathrm{\acute{a}}$cs and K$\mathrm{\ddot{o}}$rner \cite{pro_GW8}. Ahlswede and K$\mathrm{\ddot{o}}$rner gave an alternative characterization,  
		 which can be characterized over the GW network as the maximum over all shared rates when the two receive rates are set to their respective minimal, i.e., $R_{12} + R_1 = H(X)$ and $R_{12} + R_2 = H(Y)$ (refer to \cite{pro_GW9}). The lossy generalization of  G$\mathrm{\acute{a}}$cs-K$\mathrm{\ddot{o}}$rner CI at distortions $D_1$
 and $D_2$ is defined as $sup~ R_{12}$ subject to 	$R_{12} + R_1 = R(D_1)$ and $R_{12} + R_2 = R(D_2)$.	
 
\subsection{Laplacian Sources}

In many practical applications, multimedia sources are modeled by the Laplacian distribution,
\[
f_L(x)=\frac{\lambda}{2} e^{-\lambda  |x|},
\] where $\lambda$ is Laplacian parameter.

Hence considerable attention has been focused on its optimal quantization, which is discussed in the following subsections. 
\subsubsection{Efficient Scalar Quantization of Laplacian Sources}
In \cite{lap}, the optimal entropy constrained quantizer for the Laplacian source was derived to be the dead-zone plus uniform threshold quantization classification rule and the nearly uniform reconstruction rule (as illustrated in Fig.~\ref{fig:lap_sca}). This dead-zone quantizer (DZQ) has uniform step size in all of the intervals, except the dead-zone interval around zero, which is wider than the other intervals.
\begin{figure}[t]
	\centering
	\centerline{\includegraphics[width=1.1\linewidth]{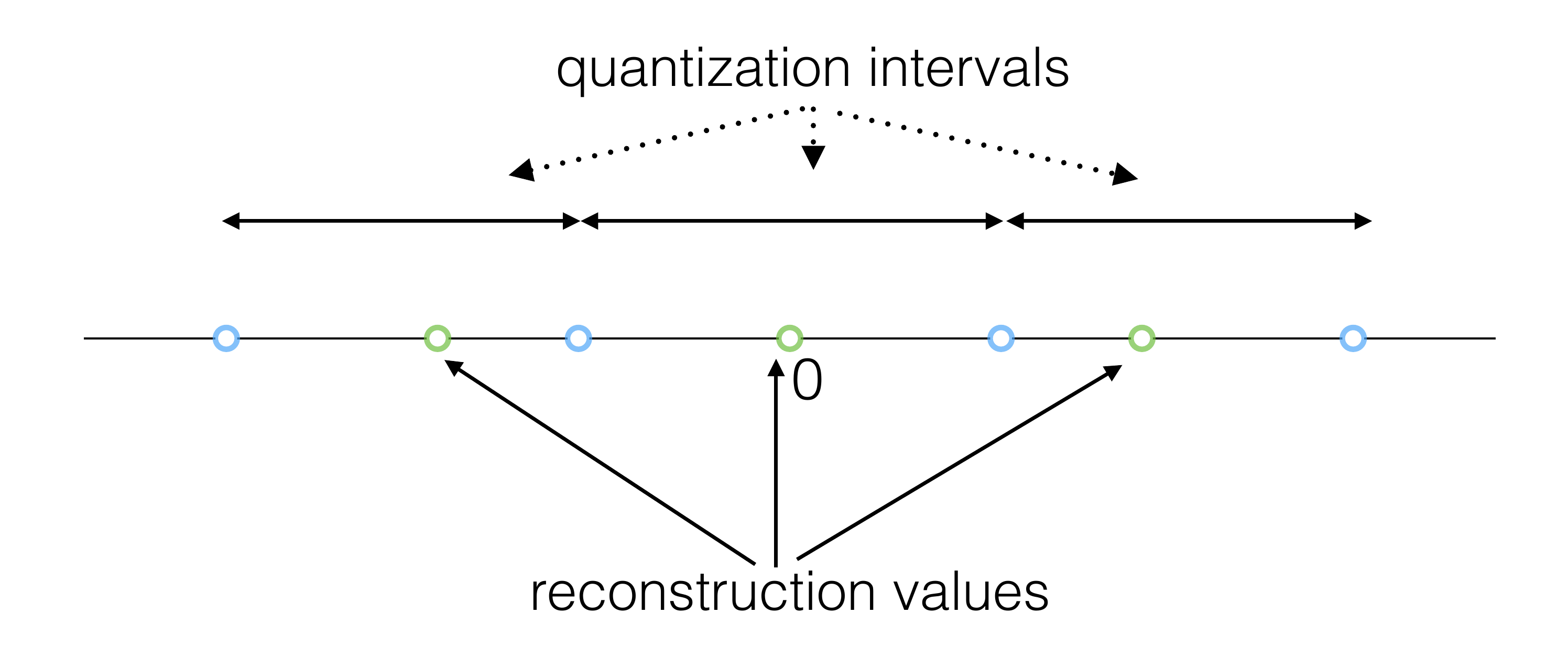}}
	\caption{ The optimal scalar dead-zone quantizer for Laplacian sources with nearly uniform reconstruction rule.}
	\label{fig:lap_sca}
\end{figure}

\subsubsection{Scalable Coding of Laplacian Sources:\\}

\textbf{In Current Multimedia Standards}

In current scalable coding standards such as, scalable HEVC \cite{HEVC} for video, and scalable AAC \cite{ISOaac} for audio, the base layer employs DZQ for quantizing the source. Then, in the enhancement layer, a scaled version of the base layer DZQ quantizes the base layer reconstruction error. 

\textbf{Conditional Enhancement Layer Quantization (CELQ)}

In \cite{CELQ}, an efficient approach for scalable coding of Laplacian sources is proposed, wherein:
\begin{itemize}
	\item The base layer employs a DZQ.
	\item The enhancement layer quantizers are conditioned on the base layer quantization interval: Use DZQ if a dead zone interval was established by the base layer, and use a uniform quantizer otherwise (as illustrated in Fig.~\ref{fig:CELQ}).
\end{itemize}
\begin{figure}[t]
	\centering
	\centerline{\includegraphics[width=1.2\linewidth]{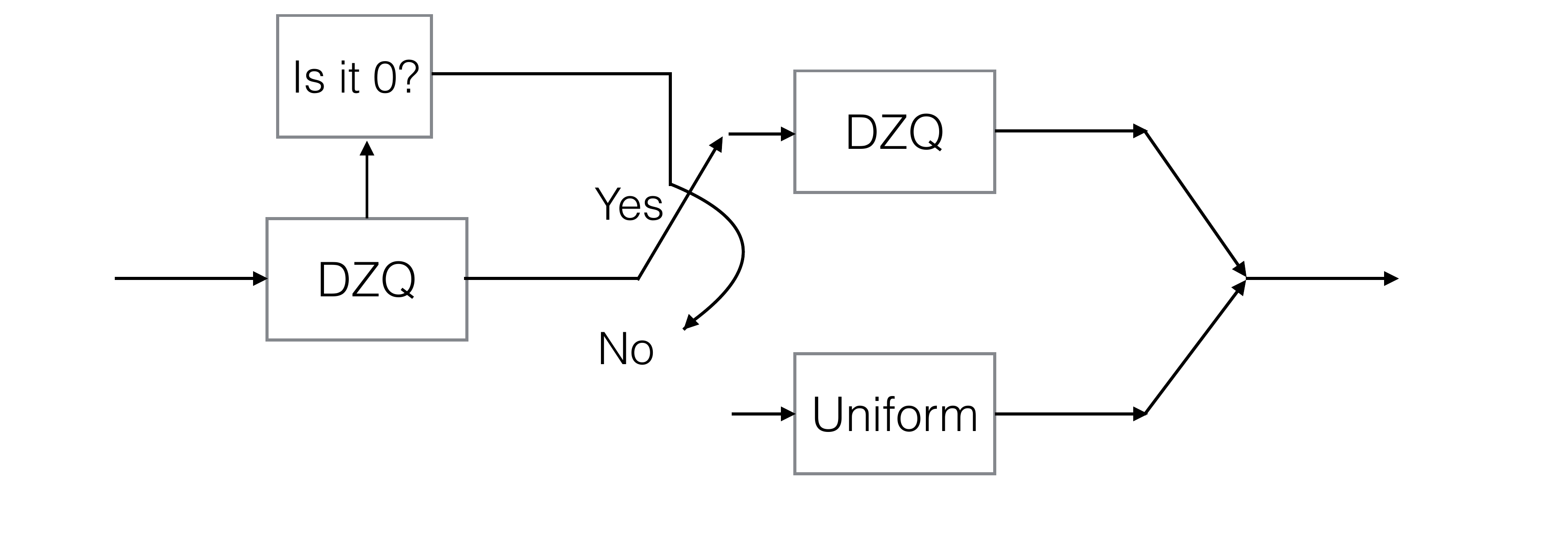}}
	\caption{Conditional enhancement layer quantizer for Laplacian sources. Based on the base layer DZQ interval, the enhancement layer quantizer is chosen.}
	
	\label{fig:CELQ}
\end{figure}
This improved scalable coder still suffers from performance penalty compared to non-scalable coding.

\end{section}
\section{Overview of the Proposed Framework}


	The proposed novel layered coding framework arises out of our lab's lossy generalization of GK common information as a special case of $X=Y$ in the GW network. In \cite{CI_NSR} we derived the information theoretic characterization of the CI between different quality levels of the same source as $sup ~ I(X;U)$ where the supremum is over all conditional densities $P(U,\hat{X}_1,\hat{X}_2|X)$ such that we have RD optimality at $D_1$ and $D_2$ and for which following Markov chains hold:
 \[
 X \leftrightarrow \hat{X}_1 \leftrightarrow U,~~~ X \leftrightarrow \hat{X}_2 \leftrightarrow U
  \]

	In this section we illustrate both the novel layered coding paradigm \cite{CI_NSR} and the significance of effective extraction of common information, via quantizer design for a toy example involving a simple  uniform distribution.
	For a uniformly distributed random variable, as shown in \cite{uniform}, the optimal entropy constrained scalar quantizer (ECSQ), at rate $\log(N)$, where $N$ is an integer, is a uniform quantizer with $N$ levels. Fig.~\ref{fig:NS} shows the optimal partitions for a uniform random variable, $U(0,6)$, at rates $R_1=2$ and $R_2=\log(6)$, resulting in distortion $D_1$ and $D_2$, respectively. 
	Clearly, boundary points of quantizer 1, do not completely align with partitions of quantizer 2, i.e., scalable coding with these receive rates is not successively refinable, and such a hierarchical coding results in an enhancement layer with distortion worse than independent quantization at rate $\log(6)$. This implies the information required to achieve $D_1$ is not a proper subset of the information required to achieve $D_2$. 
	\begin{figure}[t]	
		\begin{center}
			
			\includegraphics[width=0.5\textwidth]{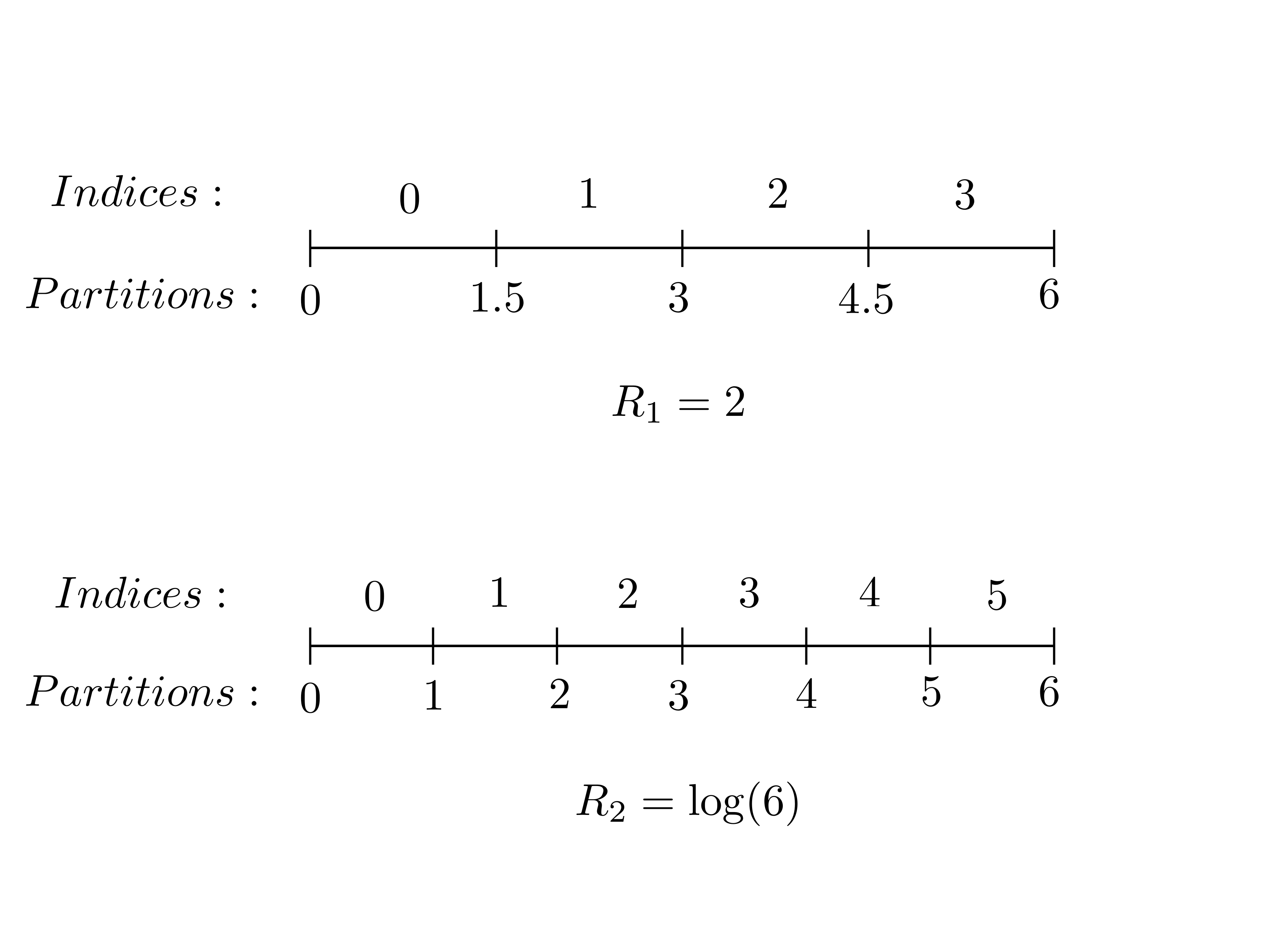}
		\end{center}
		\centering
		\vspace{-1cm}
		
		\caption{Partitions for individually quantizing a uniformly distributed random variable, $U(0,6)$,  at rates $R_1=2$ and $R_2=\log(6)$. }
		\label{fig:NS}
	\end{figure}
	\begin{figure}[t]
		\centering
		\centerline{\includegraphics[width=1\linewidth]{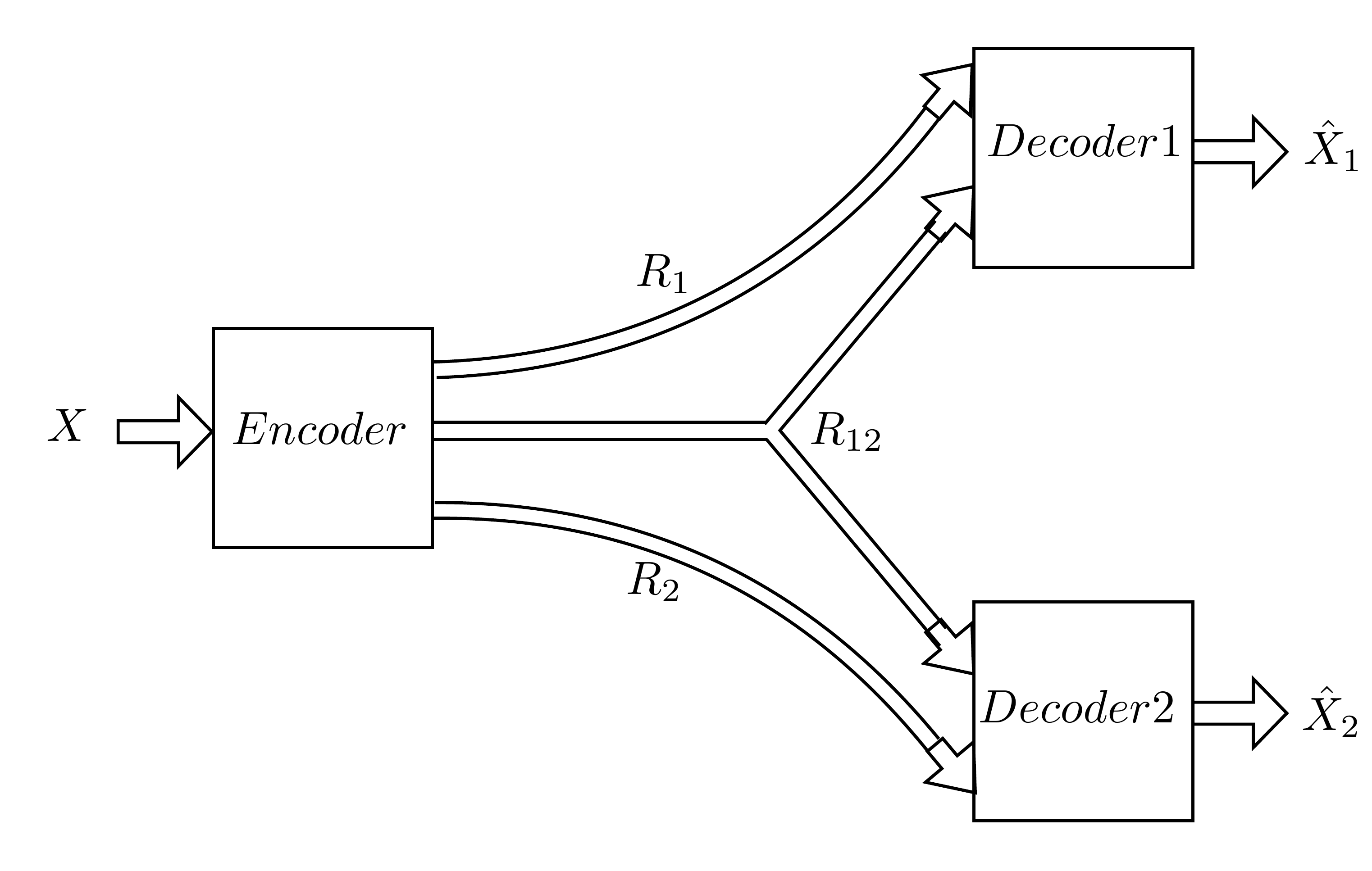}}
		\caption{Common information based layered coding paradigm with rate $R_{12}$ sent to both the decoders, and rates $R_1$ and $R_2$ sent to corresponding decoders. }
		\label{fig:wyner}
	\end{figure}
	
	\begin{figure}[t]	
		\begin{center}
			
			\includegraphics[width=0.5\textwidth]{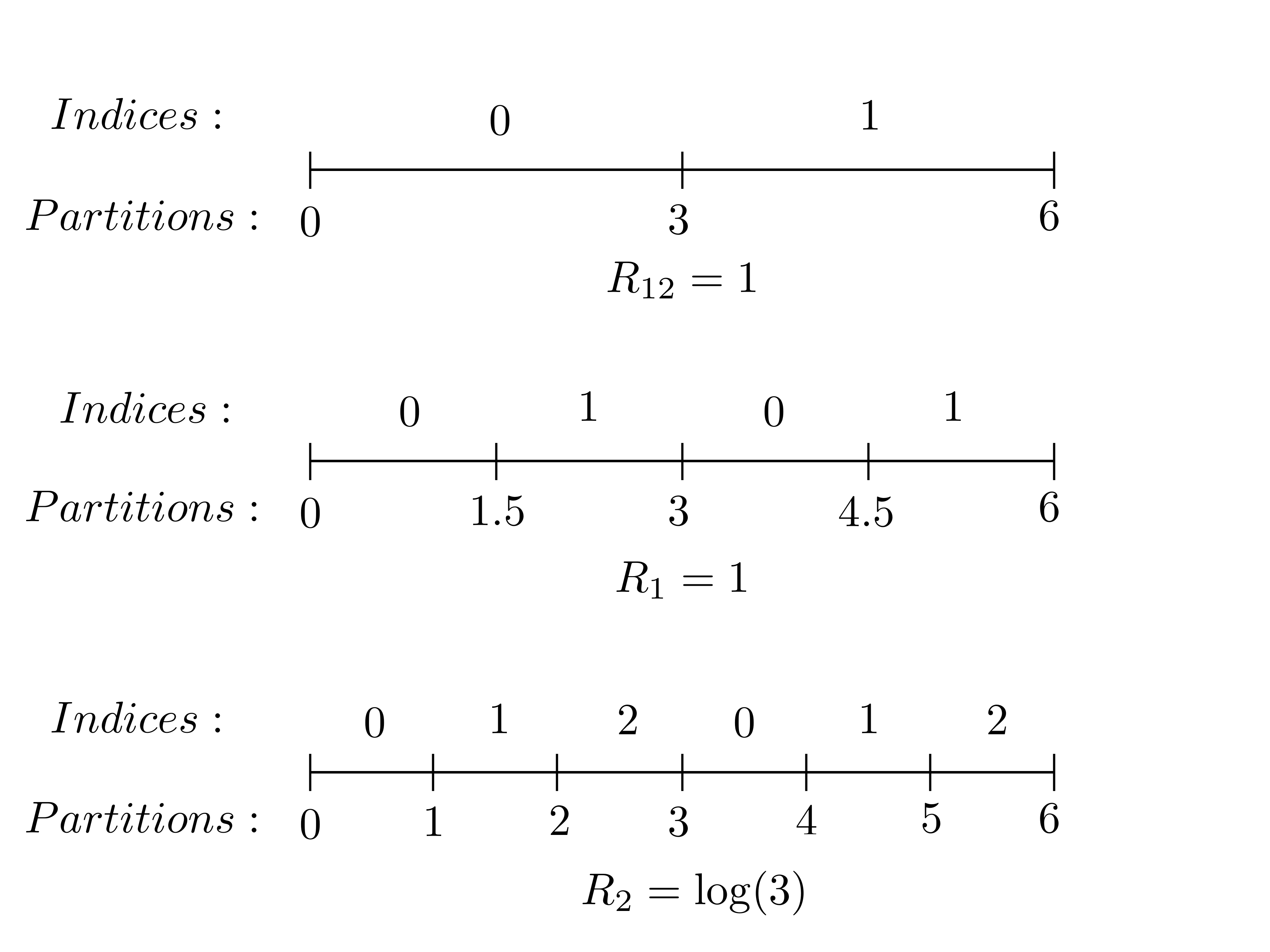}	
		\end{center}
		\centering
		\caption{Partitions for quantizing a uniformly distributed random variable, $U(0,6)$, using the common information based layered coding paradigm, with rate $R_{12}=1$ sent to both the decoders, and rates $R_1=1$ and $R_2=\log(3)$ sent to corresponding decoders.}
		\label{fig:CI}
	\end{figure}

	On the other hand independent coding is wasteful as there is obviously considerable overlap in information required to achieve the two distortions.  Instead we address this challenge with the common information based layered coding framework, wherein only part of the information required to achieve $D_1$ is sent to the decoder reconstructing at higher quality of $D_2$. This layered coding paradigm is illustrated in  Fig.~\ref{fig:wyner}, wherein three different packets are generated by the encoder at rates of:
	\begin{itemize}
		\item $R_1$, which is sent only to the decoder reconstructing at $D_1$,
		\item $R_2$, which is sent only to the decoder reconstructing at $D_2$, and
		\item $R_{12}$, which is sent to both the decoders.
	\end{itemize}
	This paradigm subsumes, as special cases, both conventional scalable coding (achieved when $R_1=0$) and non-scalable coding (achieved when $R_{12}=0$).
	Moreover, this framework provides an extra degree of freedom with which rate-distortion optimality at both the layers can be potentially achieved at a lower total transmit rate than non-scalable coding by appropriately designing the quantizers for all the layers.
	Continuing our toy example of quantizing a uniformly distributed source at receive rates of 2 and $\log(6)$, employing the layered coding paradigm, with rate $R_{12}=1$, which is sent to both decoders, and rates $R_1=1$ and $R_2=\log(3)$, which are sent to corresponding decoders (as depicted in Fig.~\ref{fig:CI}), achieves overall quantizers with partitions same as the optimal individual quantizers, ensuring rate-distortion optimality at the decoders, but with a 22\% reduction in total transmit rate compared to non-scalable coding. This example clearly demonstrates that with appropriately designed quantizers the proposed layered coding paradigm can efficiently extract information common to different quality levels. Fig.~\ref{fig:NS_CI_New}  depicts the same notion for the vector cases. 
	
	With this motivation, next we illustrate joint design of scalar and vector quantizers,  we propose a technique for low complexity design for Laplacian sources,  and finally we propose an iterative approach for joint design of vector quantizers for multi-layers.  
	

	\begin{figure*}[t]
		
		\begin{center}
			
			\includegraphics[width=1\textwidth]{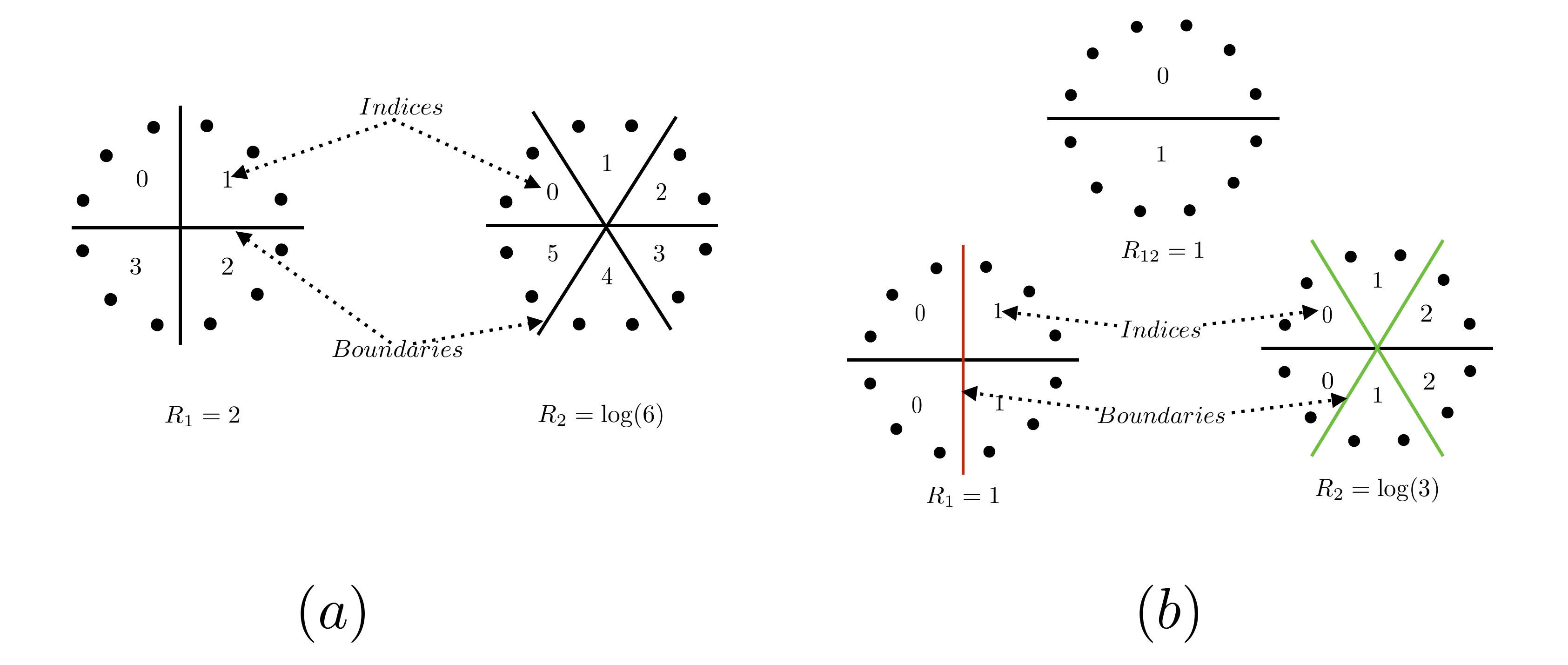}
			
		\end{center}
		
		\centering
		\caption{Partitions for quantizing a 2D training set shown above as large black dots. (a) depicts individual partitions at rates $R_1=2$ and $R_2=\log(6)$. (b) depicts partitions for the common information based layered coding paradigm, with rate $R_{12}=1$ sent to both the decoders, and rates $R_1=1$ (via red partitions) and $R_2=\log(3)$ (via green partitions) sent to corresponding decoders.}
		\label{fig:NS_CI_New}
	\end{figure*}


\begin{figure*}[t]
		\centering
		\centerline{\includegraphics[width=1\linewidth]{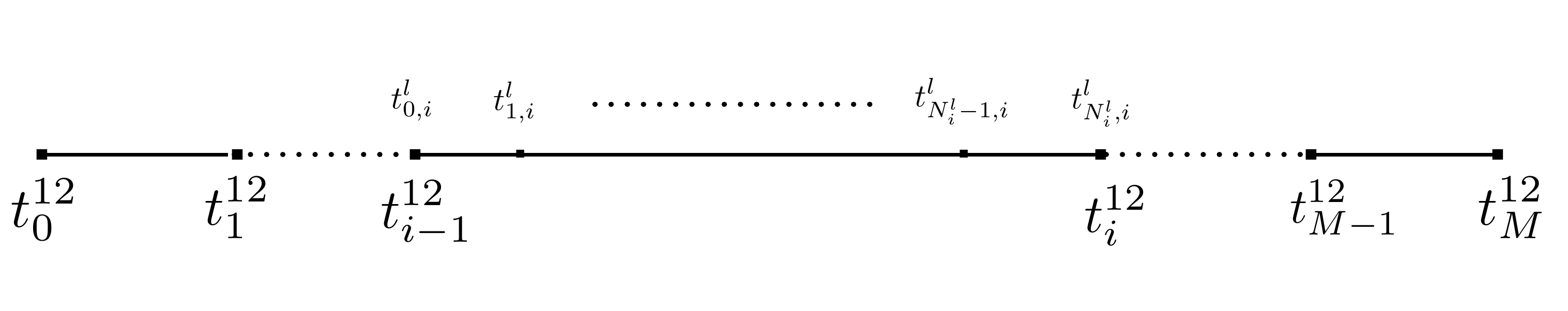}}
		\caption{Quantizer partition boundaries for subdividing a common layer quantization interval. The decision points of the common layer are shown below the line, and above it for individual layer $l$. }
		\label{fig:partition}
	\end{figure*}

\begin{section}{Joint Design of Quantizers for Common Information Based Layered Coding}
	We are primarily interested in designing Entropy Constrained Scalar/Vector Quantizer for fixed received rates, $R_{r_1}=R_{12}+R_1 = c_1$ and $R_{r_2}=R_{12}+R_2 = c_2$, at decoder $1$ and $2$, respectively. This constraint translates into a trade off between total transmit rate, $R_t= R_{12}+R_1+R_2=c_1+c_2-R_{12}$ (or equivalently $R_{12}$), and the distortions $D_1$ and $D_2$ at the decoders, with the two extremes of: i)~non-scalable coding, incurring the highest $R_t =c_1+c_2$ (or equivalently  lowest $R_{12}=0$), but yielding the lowest distortion $D^*(c_1)+D^*(c_2)$, where $D^*(\cdot)$ is the optimal distortion at a given rate; and ii)~conventional scalable coding, with the lowest $R_t = c_2$ (or equivalently  highest $R_{12}=c_1$), but yielding high distortion at the enhancement layer due to the scalable coding penalty. To design our layered coding quantizers to optimize this trade off subject to prescribed receive rate constraints, we define our cost function as,
	\begin{equation}
		\label{eqn:j}
	J=  a_1 D_1 + a_2 D_2 + \lambda_1(R_1+R_{12})+\lambda_2(R_2+R_{12})+\lambda_{12}R_{12},
	\end{equation}
	
	where $a_1$ and $a_2$ control the importance of $D_1$ and $D_2$, respectively, and $\lambda_1$, $\lambda_2$, and $\lambda_{12}$ constrain $R_{r_1}$, $R_{r_2}$ and $R_{12}$, respectively.  Also note that one of the five weights, $a_1, a_2, \lambda_1, \lambda_2, \text{and } \lambda_{12}$, in the cost function \eqref{eqn:j} is obviously redundant, but we decided to keep this harmlessly redundant notation for consistency with past literature. Quantizers designed to minimize the cost function \eqref{eqn:j} achieve the best weighted sum of distortions at the decoders (as per $a_1,a_2$), for a given total transmit rate (or equivalently $R_{12}$) and receive rates.
	
	\subsection{Quantizer Design for Scalar Sources:}
	 We design the quantizers iteratively, with one quantizer updated in each iteration to minimize the overall cost function while the others are fixed, and repeating the iterations until convergence.
	In the following subsections, superscripts ``1'', ``2'' and ``12'' refer to parameters of individual layers $1$,  $2$, and common layer, respectively. 
	
	\subsubsection{Individual Layer Quantizer Design}
	\label{sec:indq}
	Given a common layer quantizer with $M$ intervals and partition boundaries $ t_i^{12},~i= 0,1, \cdots, M$, we need to design ECSQ for each interval, $(t_{i-1}^{12},t_i^{12}),~i= 1,2, \cdots, M$ and each individual layer, $l=1,2$. Note that, given the common layer, the individual quantizers of layers $1$ and $2$ have absolutely no effect on each other's distortion and rate and can be optimized independently. Hence the cost function for optimization of individual quantizers for each of the layer simplifies to $ J_l=a_l D_l+\lambda_l R_l$ for $l=1,2$.
	
	We employ the well known iterative ECSQ design technique for each interval, $(t_{i-1}^{12},t_i^{12}),~i= 1,2, \cdots, M$ and each individual layer, $l=1,2$. Let $N_i^l$ be the number of subintervals for layer $l$ at common layer interval  $(t_{i-1}^{12},t_i^{12})$. Partition boundaries for layer $l$ at interval $(t_{i-1}^{12},t_i^{12})$ are shown in Fig.~\ref{fig:partition} as $t_{q,i}^l,~q= 0,1,2, \cdots, N_i^l$, and we can clearly see that $t_{0,i}^l=t_{i-1}^{12}$ and $t_{N_i^l,i}^l=t_{i}^{12}$. 
	The iterative ECSQ algorithm employed is described below:
	\begin {enumerate}
	\item Initialize the partition boundaries, $t_{q,i}^l,~q= 1,2, \cdots, N_i^l-1$.
	\item Calculate $N_i^l$ representative levels $x_{q,i}^l$ and  subinterval probabilities $p_{q,i}^l$ for $q=1,2, \cdots, N_i^l$:
	\begin{equation}
	 x_{q,i}^l = \frac{\int_{t_{q-1,i}^l}^{t_{q,i}^l} x f(x) dx}{\int_{t_{q-1,i}^l}^{t_{q,i}^l}  f(x) dx} 
	\end{equation}
	
	\begin{equation}
	 p_{q,i}^l=\int_{t_{q-1,i}^l}^{t_{q,i}^l}  f(x) dx , ~~~ q=1,2, \cdots, N_i^l,
	\end{equation}
	where $f(x)$ denotes the source distribution. 
	\item Calculate partition boundaries, $t_{q,i}^l; ~ q= 1,2, \cdots, N_i^l-1$:
	\begin{equation}
	 t_{q,i}^l = \frac{x_{q,i}^l +x_{q+1,i}^l }{2}  - \frac{\lambda_l}{a_l} * \frac{\log_2 p_{q,i}^l - \log_2 p_{q+1,i}^l}{2(x_{q,i}^l -x_{q+1,i}^l )},
	\end{equation}
	\item Repeat $(2)$ and $(3)$ until there is no further reduction in cost (or a prescribed stopping criterion is met).
\end{enumerate}

\subsubsection{Common Layer Quantizer Design} 
\label{sec:cq}
Unlike the individual layer quantizer design, any partition change in the common layer quantizer has impact on all the distortions and rates, hence the standard ECSQ update rules are not applicable. Given quantizer partition boundaries and representative levels at both individual layers, optimal partition points at common layer, $ t_i^{12},~i= 1,2, \cdots, M-1$,  (as derived in Appendix~\ref{apn:cpart})  is updated as: 
\begin{dmath}
	t_i^{12} = \frac{(a_1 (x_{1,i+1}^1)^2+a_2 (x_{1,i+1}^2)^2)-(a_1 (x_{N_i^1,i}^1)^2+a_2 (x_{N_i^2,i}^2)^2)}{2((a_1 {x_{1,i+1}^1}+a_2 {x_{1,i+1}^2})-(a_1 {x_{N_i^1,i}^1}+a_2 {x_{N_i^2,i}^2}))}  - \hspace{-1.5pt}
	\frac{\lambda_1 (\log_2 p_{1,i+1}^1 - \log_2 p_{N_i^1,i}^1) + \lambda_2 (\log_2 p_{1,i+1}^2 - \log_2 p_{N_i^2,i}^2)}{{2((a_1 {x_{1,i+1}^1}+a_2 {x_{1,i+1}^2})-(a_1 {x_{N_i^1,i}^1}+a_2 {x_{N_i^2,i}^2})})}\\
	\hspace{-1.2pt}
	-\frac{ \lambda_{12} (\log_2 p_{i+1}^{12} - \log_2 p_{i}^{12})}{{2((a_1 {x_{1,i+1}^1}+a_2 {x_{1,i+1}^2})-(a_1 {x_{N_i^1,i}^1}+a_2 {x_{N_i^2,i}^2})})},
\end{dmath}
where, $p_{i}^{12}=\int_{t_{i-1}^{12}}^{t_{i}^{12}}  f(x) dx,  ~~~ i=1,2, \cdots,M$ are based on previous values of $t_i^{12}$.

\subsubsection{Joint Design of Quantizers}
The overall algorithm for joint design of quantizers for all the layers is the following:
\begin {enumerate}
\item Initialize the partition boundaries for common layer, $t_i^{12},~ i= 1,2, \cdots, M-1$.
\item Update the individual layer quantizers using the approach described in Section~\ref{sec:indq} to calculate $x_{q,i}^l$, $p_{q,i}^l$  for $q=1,2, \cdots, N_i^l$ and $t_{q,i}^l$  for $q=1,2, \cdots, N_i^l-1$ for both $l= 1,2$ and all $ i= 1,2, \cdots, M$.
\item Update the common layer quantizer using formula given in Section~\ref{sec:cq} to calculate $t_i^{12}, ~ i= 1,2, \cdots, M-1$.
\item Repeat $(2)$ and $(3)$ until there is no further reduction in cost (or a prescribed stopping criterion is met).
\end{enumerate}
Note that during the ECSQ design at common or individual layers, the number of partitions, i.e., $M$ and $N_i^l$ are not known. To circumvent this, we simply initialize our algorithm with a large number of partitions, and in each iteration, based on the given $a_1, a_2, \lambda_1, \lambda_2,  \text{and } \lambda_{12}$, the algorithm reduces the number of partitions, if necessary. 

\subsection{Low Complexity Quantizer Design for Laplacian Sources:}
For the practically important case of Laplacian source distribution we could have the low complexity non-joint design  as below:
\begin{enumerate}
	\item For the common layer, we estimate the best step size for the DZQ at a given rate, $R_{12}$.
	\item For the two individual layers, we design optimal entropy constrained quantizers for each common layer quantizer interval, at their corresponding rates of $R_1$ and $R_2$. Specifically, we iteratively optimize the quantizer interval partitions and reconstruction points to minimize the entropy constrained distortion, with smart initializations of,\label{stp2}
	\begin {itemize}
	\item A DZQ for the dead zone interval, and 
	\item A uniform quantizer for other intervals,
	\end {itemize}
	of the common layer quantizer.
	\item We then numerically estimate the optimal common layer rate, by trying multiple allowed common rates and selecting the one that results in minimum cost $J$.
\end{enumerate}
Since the dead zone interval contains a truncated Laplacian distribution and other intervals contain a truncated exponential distribution, we select the initializations in Step~\ref{stp2} above to be the optimal entropy constrained quantizers of their corresponding non-truncated distributions.

Note that, we can achieve a non-zero common rate with negligible $\Delta D$, if the DZQ at rate $R_{12}$ is such that all its partition points align closely with partition points of DZQ at both rate $c_1$ and $c_2$. Conditions for such an alignment of partitions between two DZQ were derived in \cite{embedded} as, the dead-zone of the coarser DZQ has to be divided into $2n+1$ intervals, and other intervals of this DZQ have to be divided into $m+1$ intervals, with $2n/m=z$, where, $n$ and $m$ are integers, and $z$ is the ratio of the dead-zone interval length over other intervals' length. Our design technique numerically estimates the common layer DZQ which closely satisfies these conditions with DZQ at both rate $c_1$ and $c_2$.


Note that the proposed design technique does not ensure joint optimality of all the quantizers, particularly since we independently optimize the common layer quantizer (e.g., DZQ for Laplacian) without considering its effect on other layers. Despite this assumption we obtain considerable performance improvements. 

		\subsection{Quantizer Design for Vector Sources:} \label{sec:qdvs}

		
We design the quantizers iteratively by alternating between the steps of optimal partitioning and optimal codebook estimation, until convergence, similar to the generalized Lloyd algorithm \cite{lbg}. 
		
		
		We design an $M$-codebook ECVQ for the common layer, and for each common layer region, $i= 1,2, \cdots, M$ we design an ECVQ for each individual layer, $l=1,2$. Let $N_i^l$ be the number of subregions for layer $l$ at common layer region $i$. Let $\mathbf{c}^l_{q,i}$ and $p^l_{q,i}$ be the representative levels, and subregion probabilities, respectively, for $q= 1,2, \cdots, N_i^l$, $i= 1,2, \cdots, M$, and $l=1,2$. Finally, let $p^{12}_i$ be the common layer regions probabilities for $i= 1,2, \cdots, M$. Following is the overall iterative algorithm:
		\begin{enumerate}
			\item Guess an initial set of representative levels $\mathbf{c}^l_{q,i}$ and their corresponding probabilities $p^l_{q,i}$ for $q= 1,2, \cdots, N_i^l$, $i= 1,2, \cdots, M$, and $l=1,2$.
			\item Assign each sample $\mathbf{x}_t$ in training set $S$ to common layer region $i$ and subregions' representatives $\mathbf{c}_{q_1,i}^1$ and $\mathbf{c}_{q_2,i}^2$, to minimize the Lagrangian cost:
			\begin{align}
				J_{\mathbf{x}_t} (\mathbf{c}_{q_1,i}^1 , \mathbf{c}_{q_2,i}^2) &= (a_1||\mathbf{x}_t-\mathbf{c}_{q_1,i}^1||^2 +  a_2||\mathbf{x}_t-\mathbf{c}_{q_2,i}^2||^2 ) - \nonumber \\
				& (\lambda_1log_2 p^1_{q_1,i} + \lambda_2log_2 p^2_{q_2,i} +  \lambda_{12}log_2 p_i^{12})
			\end{align}
			\item Find subregion $B^l_{q,i}$:
			
			$B^l_{q,i}=\{\mathbf{x} \in S: \mathbf{x} \text{ is assigned to }\mathbf{c}^l_{q,i} \}$
			
			for $q= 1,2, \cdots, N_i^l$, $i= 1,2, \cdots, M$, and $l=1,2$.
			\item Calculate new representative levels and probabilities:
			\begin{eqnarray}
				\mathbf{c}_{q,i}^l&=&\frac{1}{||B^l_{q,i}||} \sum_{\mathbf{x}\in B^l_{q,i}}\mathbf{x},\\
				p_{q,i}^l&=&\frac{||B^l_{q,i}||}{||S||},
			\end{eqnarray}
			for $q= 1,2, \cdots, N_i^l$, $i= 1,2, \cdots, M$, and $l=1,2$.
			Also $p_i^{12}=\sum\limits_{q=1}^{N_i^1}p_{q,i}^1$ for $i= 1,2, \cdots, M$
			\item Repeat steps 2, 3, and 4 until there is no further reduction in overall cost:
			$J=\sum\limits_{\mathbf{x}_t} J_{\mathbf{x}_t}$
		\end{enumerate}

		\textbf{Notes on the designed layered coding quantizers:}
		
		The common layer quantizer does not reconstruct the source, leading to flexibility in structure of the quantizer regions. That is, we could potentially have irregular quantizers in the common layer alone, or both common and individual layers, which combine together to result in a overall regular quantizer. Surprisingly, despite this flexibility, we show that the following lemma on regularity of the common layer quantizer is true with one simple assumption.
		\begin{lemma}
			If the final overall quantization cells are regular, then so are the optimal common layer quantization cells.
		\end{lemma}
		\begin{proof}
			Assume common layer's region $j$ is not regular. Without loss of generality, we assume it consists of two regular subregions $j_1$ and $j_2$.
			Let us define a new common layer quantizer partition by dividing region $j$ into two regions $j_1$ and $j_2$.  
			Since we assumed the final overall quantization regions are regular, subregions of $j$ in individual layers are forced to be wholly contained within either region $j_1$ or $j_2$. Which implies that dividing common layer region $j$ into two regions $j_1$ and $j_2$, does not alter the final overall quantization regions. Since the distortions at the decoder depend only on the final overall quantization regions, they do not change when the common layer region $j$ is divided into two regions $j_1$ and $j_2$. Also note that, receive rate at decoder $l$ is:
			\[R_l+R_{12} = - \sum_{i=1}^{M} \sum_{q=1}^{N_i^l} p_{q,i}^l \log_2 p_{q,i}^l , ~~~\text{for} ~~l=1,2\]
			Since the final overall quantization regions are not altered, rate contributions from subregions of common layer region $j$ (via $p_{q,j}^l$)
			gets redistributed to rate contributions from subregions of common layer region $j_1$ (via $p_{q,j_1}^l$) and region $j_2$ (via $p_{q,j_2}^l$), and thus $R_l+R_{12}$ for $l=1,2$, do not change.
			However, we clearly increase $R_{12} = -\sum_{i=1}^M p_i^{12}  \log_2  p_i^{12}$ by dividing the $j$th region into two regions. That is, we obtain a new set of quantizers with the same distortions and received rates as before, but lower transmit rate (or equivalently higher $R_{12}$), which means the cost $J$ is reduced. Hence, the optimal common layer quantizer cannot be irregular.
		\end{proof}

		\subsection{Multiple Quality Level Layered Coding}
		It is not straightforward to extend the concept of common information to more than two quality levels. The biggest challenge is because shared information can exist between any of the levels. Hence clearly, just one common layer between all the quality levels cannot exploit all the redundancies present. In fact, common information can exist between any subset of quality levels, leading to number of bit-stream layers growing combinatorially with number of quality levels.
Our lab has recently demonstrated that this ``combinatorial message sharing'' can be used to improve theoretically achievable region for the related problem of multiple descriptions coding \cite{pro27}. While this approach is useful to obtain asymptotic bounds, it is intractable for real world applications to maintain such large number of bit-stream layers.
Hence we propose employing a linearly growing rate-splitting approach where each layer receives an individual packet for itself and all the common packets received by lower layers (layers with higher distortion constraints). Specifically, with $L$ decoders, we plan to generate $2L-1$ packets consisting of $L$ individual packets at rates $R_1$ to $R_L$ and $L-1$ common packets at rates $R_{123\cdots R_L}$ to $R_{(L-1)L}$ send to decoders as following:
		\begin{align}
		Q = \{& 1 \\ \nonumber
		& 2, \\ \nonumber
		& \quad .  \\  \nonumber
		& \quad . \\  \nonumber
		& \quad . \\  \nonumber
		& L-1, \\ \nonumber
		& L, \\ \nonumber	
		& 123\cdots L,  \\ \nonumber
		& 23 \cdots L,  \\  \nonumber
		& \quad .  \\  \nonumber
		& \quad . \\  \nonumber
		& \quad . \\  \nonumber
		& \left( L-1 \right)L\}
		\end{align}
		where for example $23 \cdots L$ means the packet sent to decoders $2$ till decoder $L$. Fig.~ \ref{fig:multi} depicts this scenario.
		
		\begin{figure}[t]
			\centering
			\centerline{\includegraphics[width=1\linewidth]{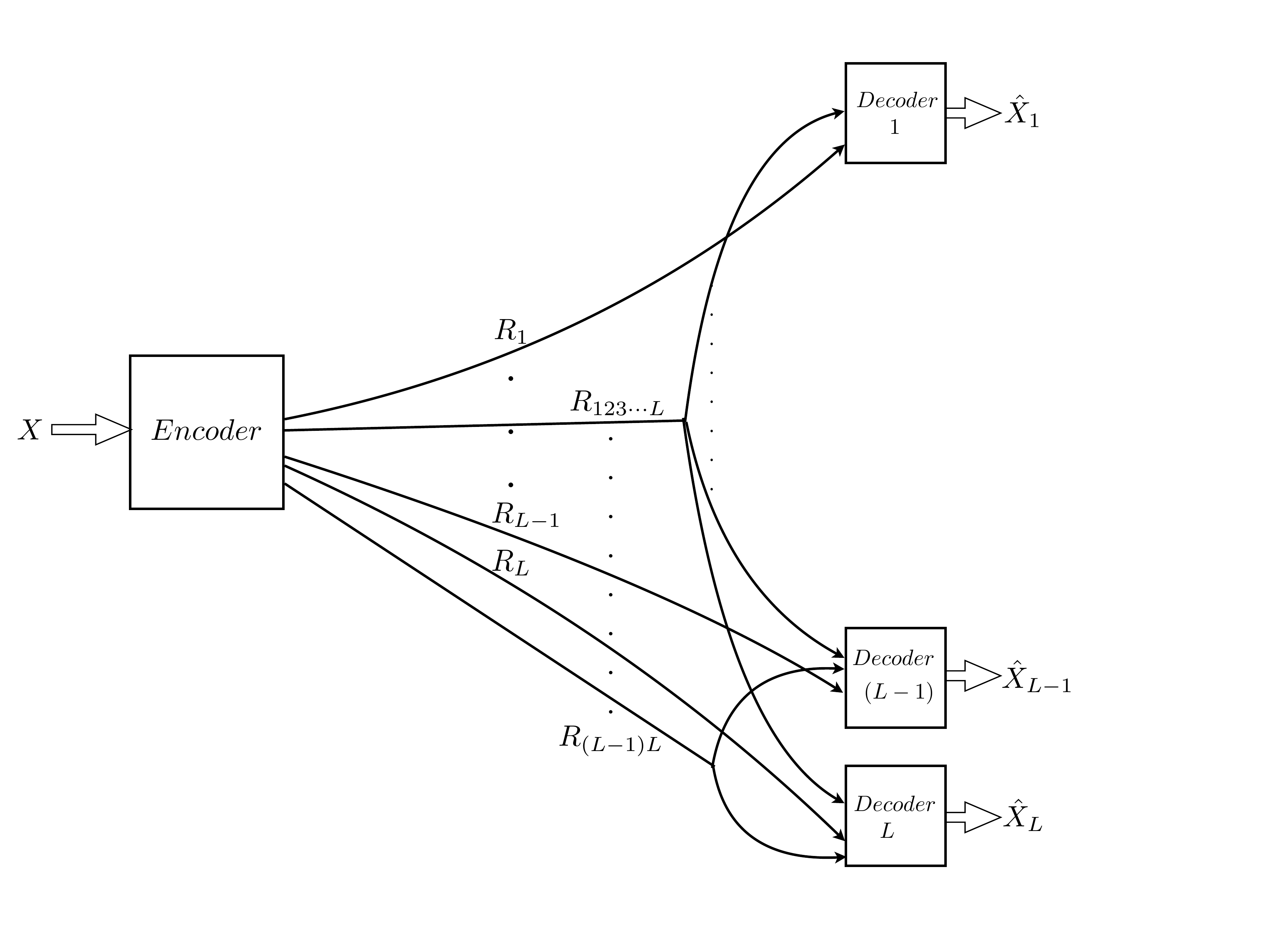}}
			\caption{Common information based layered coding paradigm for multi-layer scenario. }
			\label{fig:multi}
		\end{figure}
		
		Similar to previous part we could define the cost as 
		\begin{align}
		J = \sum_{l=1}^L a_lD_l + \sum_{q\in Q}\lambda _q R_q \label{eqn:multi-cost}
		\end{align}
		where the first and second terms in the cost $J$, represents the distortion and rate penalties respectively. Finally using the approach explained in Section \ref{sec:qdvs} we design the best entropy constrained vector quantizers. For the specific example of 3 quality level coding, the step 2) of algorithm in Section \ref{sec:qdvs}, will assign each sample $\mathbf{x}_t$ in training set $S$ to region $i$ of common layer shared between all three levels (that is sent at rate $R_{123}$), subregion $(q_1,i)$ within $i$ for private layer of quality level 1 (that is sent at rate $R_1$), subregion $(q_{23},i)$ within $i$ for common layer shared between quality level 2 and 3 (that is sent at rate $R_{23}$), subregion $(q_2,q_{23},i)$ within $(q_{23},i)$ for private layer of quality level 2 (that is sent at rate $R_2$) and subregion $(q_3,q_{23},i)$ within $(q_{23},i)$ for private layer of quality level 3 (that is sent at rate $R_3$), to minimize the cost in \eqref{eqn:multi-cost}. Other steps get extended to these many layers appropriately.
	\end{section}

	\begin{section}{Experimental Results}
		\label{sec:results}

	\subsection{Results for the Joint Design of Scalar Quantizers:}
		\label{sec:results}
		We provide experimental results for Laplacian sources and set $\lambda=1$ and measure the distortions at each decoder in dB.

		In our first experiment, we fixed the receive rates to $c_1=2$ and $c_2=3$. Fig.~\ref{fig:lap} shows the total transmit/storage rate $R_t$ versus  the excess distortion  $\Delta D = D_1+D_2- D^*(c_1)-D^*(c_2)$ curve obtained by employing quantizers designed by our proposed iterative technique at various common layer rates ($R_{12}$) ranging from 0 bits (i.e., non-scalable coding) to 2 bits (i.e., scalable coding), and the convex hull between non-scalable and scalable coding, which is obtained using the time sharing argument. Note that scalable coding employed by current standards is around 1.5 dB worse than the efficient scalable coding subsumed by our approach, which itself has around 0.8 dB distortion loss compared to non-scalable coding. The results clearly demonstrate that, by exploiting the concept of common information, the proposed approach can achieve all intermediate operating points at considerably better performance compared to the convex hull between non-scalable and scalable coding.
		
		\begin{figure}[t]
			\centering
			\centerline{\includegraphics[width=.5\textwidth]{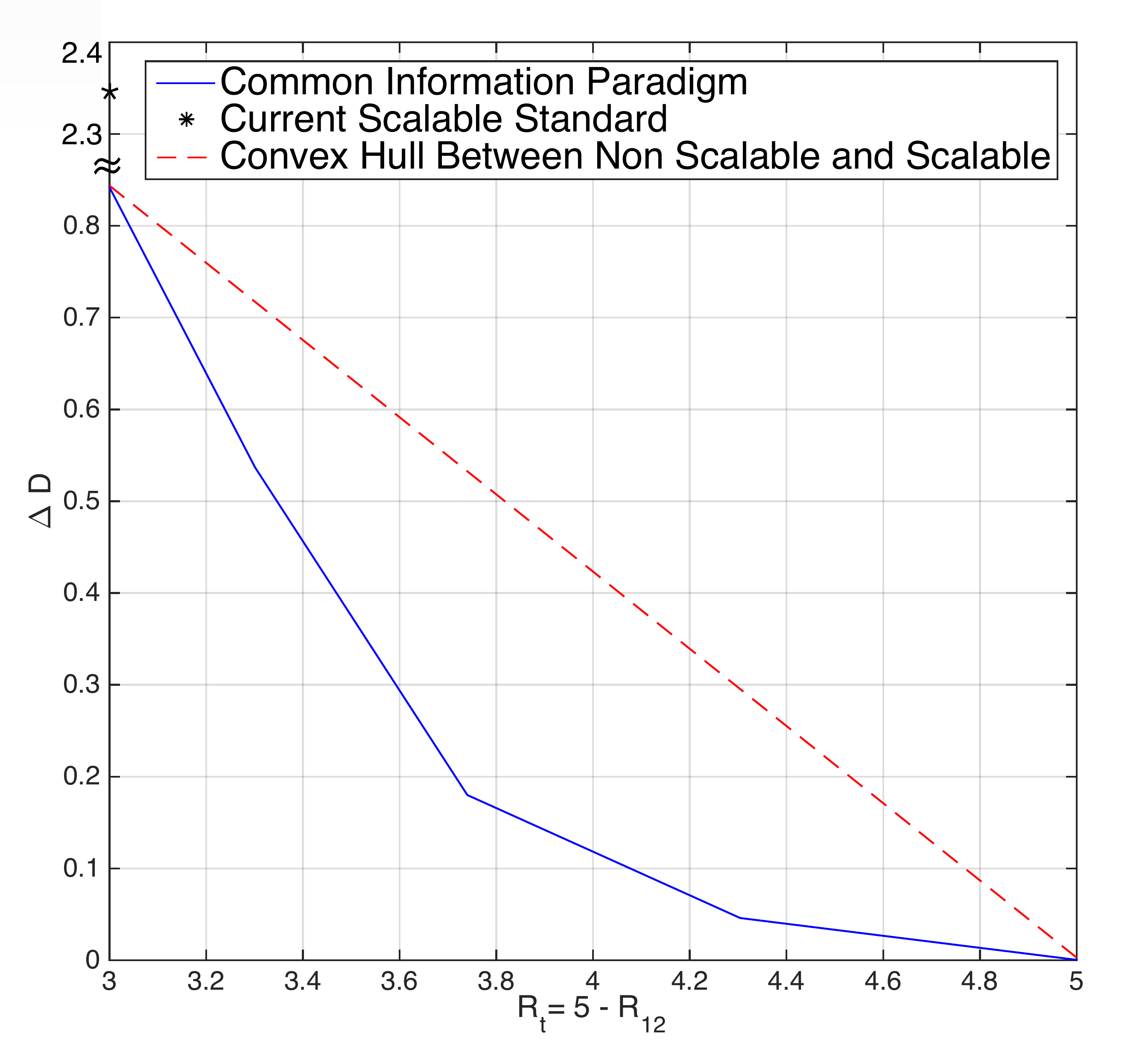}}
			\caption{Total transmit rate vs distortion deviations at the decoders. }
			\label{fig:lap}
		\end{figure}
		\begin{table*}[t]
			\centering
			\scalebox{1}{
				\begin{tabular} {|c|c|c|c|}
					\hline
					&Non scalable total transmit rate & Proposed method total transmit rate &  Total transmit rate reduction \\ 
					&$R_{12}+R_1+R_2= R_t^{NS}$ & $R_{12}+R_1+R_2=R_t^P$ & $(R_t^{NS}-R_t^P)/R_t^{NS} $\\ \hline
					$R_{r_1}=2, R_{r_2}=4$&$0 + 2 + 4= 6 $& $0.8 + 1.2+ 3.2 = 5.2$ &  $(6-5.2)/6= 13$\% \\ \hline
					$R_{r_1}=2.1, R_{r_2}=3.2$&$0 + 2.1 + 3.2 = 5.3 $& $0.8 + 1.3+ 2.4 = 4.5$ &  $(5.3-4.5)/5.3= 15$\% \\ \hline
					$R_{r_1}=2.5 , R_{r_2}=4.3$&$0 + 2.5+ 4.3 = 6.8 $& $1.1 + 1.4+ 3.2= 5.7$ &  $(6.8-5.7)/6.8= 16 $\% \\ \hline
					$R_{r_1}=2.5 , R_{r_2}=3.2$&$0 + 2.5 + 3.2= 5.7 $& $1.2 + 1.3+ 2 = 4.5$ &  $(5.7-4.5)/5.7= 21 $\% \\ \hline
					$R_{r_1}=3 , R_{r_2}=3.6$&$0 + 3 + 3.6= 6.6 $& $1.5 + 1.5+ 2.1 = 5.1$ &  $(6.6-5.1)/6.6= 23 $\% \\ \hline
					$R_{r_1}=2.7, R_{r_2}=3.4$&$0 + 2.7 + 3.4= 6.1 $& $1.4 + 1.3+ 2 = 4.7$ &  $(6.1-4.7)/6.1= 23 $\% \\ \hline
					$R_{r_1}=2.4 , R_{r_2}=3$&$0 + 2.4+ 3 = 5.4 $& $1.5 + 0.9+ 1.5= 3.9$ &  $(5.4-3.9)/5.4= 28 $\% \\ \hline
					$R_{r_1}=3.1 , R_{r_2}=3.6$&$0 + 3.1+ 3.6= 6.7 $& $2 + 1.1+ 1.6 = 4.7$ &  $(6.7-4.7)/6.7= 30 $\% \\ \hline

				\end{tabular}}
				\caption{Total transmit rate for non-scalable coding and proposed paradigm operating with negligible loss in distortion, for joint design approach.}
				\label{tab:re}
			\end{table*}
			
			Also as seen in Fig.~\ref{fig:lap}, we obtain an interesting operating point at $R_t= 4.3$ (or $R_{12}=0.7$), with distortions very close to non-scalable coding, but with a $14\%$ reduction in total transmit rate compared to non-scalable coding. We thus conducted more experiments with different fixed receive rate constraints and tabulated the corresponding results of transmit rate savings with negligible distortion loss in Table~\ref{tab:re}. The significant transmit rate savings ranging from 13\% to 30\% clearly demonstrate the capability of the proposed technique to efficiently extract the common information between different quality levels. These savings translate to significant operating cost reduction at data centers for storage, and transmitting to and caching at intermediate nodes for content distributors who are currently coding individual non-scalable copies at different quality levels.
			
			\subsection{Results for Low Complexity Quantizer Design for Laplacian Sources:}		
			In our experiments we used a Laplacian source with $\lambda=1$ and the distortions at each decoder are measured in dB. For our first experiment we used fixed receive rates of $c_1=1.6$ and $c_2=2.8$. In Fig.~\ref{fig:lap0} we plot, the $R_t$ versus $\Delta D$ curve obtained by employing quantizers designed by our proposed technique at various common layer rates ($R_{12}$) ranging from 0 bits (i.e., non-scalable coding) to 1.6 bits (i.e., scalable coding using CELQ), and the convex hull for the proposed paradigm, which is obtained using the time sharing argument. Note that scalable coding employed in current standards has around 1.5 dB distortion loss compared to efficient scalable coding, which itself has around 0.8 dB distortion loss compared to non-scalable coding. The proposed technique can operate at all points along the convex hull and at considerably better performance compared to the scalable coding of current standards.
			
%

			\begin{figure}[t]
				\centering
				\centerline{\includegraphics[width=.5\textwidth]{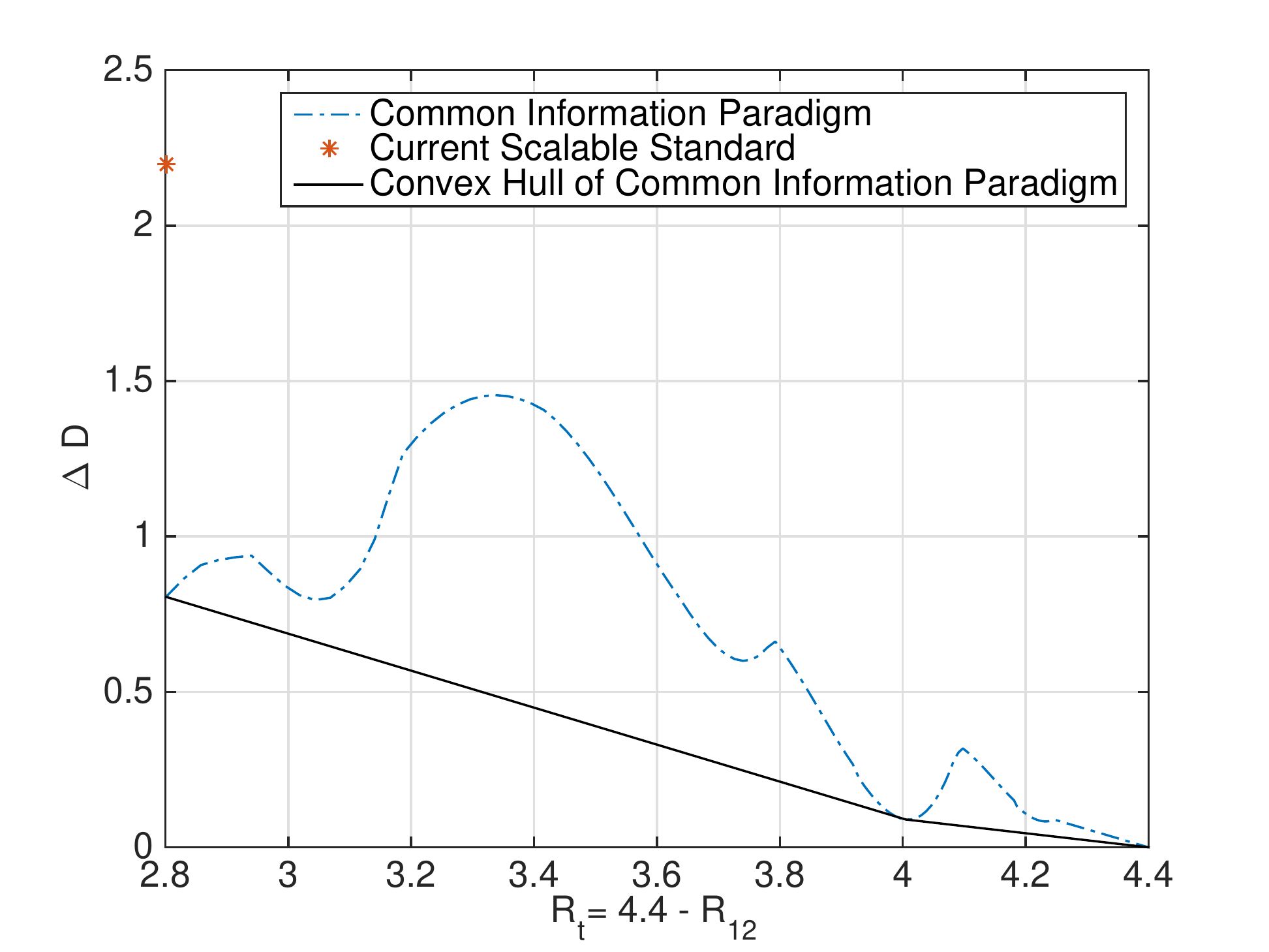}}
				\caption{Total transmit rate vs distortion deviations at the decoders. 
				}
				\label{fig:lap0}
			\end{figure}
			
			Moreover, note that in Fig.~\ref{fig:lap0} at $R_t= 4$ or equivalently $R_{12}=0.4$, we obtain distortions that are very close to that of non-scalable coding at a $9 \%$  reduction in total transmit rate compared to that of non-scalable coding. We conducted another experiment with different fixed receive rate combinations and obtained similar results of transmit rate savings with negligible distortion loss, which are shown in Table~\ref{tab:re111} to demonstrate the capability of the proposed technique to efficiently extract the common information between different quality levels.
			\begin{table*}[t]
				\centering
				\scalebox{1}{
					\begin{tabular} {|c|c|c|c|}
						\hline
						&Non scalable total transmit rate & Proposed method total transmit rate &  Total transmit rate reduction \\ 
						&$R_{12}+R_1+R_2= R_t^{NS}$ & $R_{12}+R_1+R_2=R_t^P$ & $(R_t^{NS}-R_t^P)/R_t^{NS} $\\ \hline
						$R_{r_1}=1.6 , R_{r_2}=2.8$&$0 + 1.6 + 2.8 = 4.4 $& $0.4 + 1.2+ 2.4 = 4$ &  $(4.4-4)/4.4= 9 $\% \\ \hline
						$R_{r_1}=1.5 , R_{r_2}=2.3$&$0 + 1.5+ 2.3 = 3.8 $& $0.3 + 1.2+ 2= 3.5$ &  $(3.8-3.5)/3.8= 8 $\% \\ \hline
						$R_{r_1}=1.4 , R_{r_2}=2$&$0 + 1.4 + 2= 3.4 $& $0.3 + 1.1+ 1.7 = 3.1$ &  $(3.4-3.1)/3.4= 9 $\% \\ \hline
					\end{tabular}}
					\caption{Total transmit rate for non-scalable coding and proposed paradigm operating with negligible loss in distortion for low complexity approach.}
					\label{tab:re111}
				\end{table*}
				
				\textbf{Note:} Comparison between the joint design and low complexity design shows that, joint design has better ability to extract more common information, however low complexity approach is faster to extract the common information.  			
			
			\subsection{Results for the Joint Design of Vector Quantizers:}
			\label{sec:res}

			In our next experiment we use a multivariate normal distribution $x \sim N(\mu,\Sigma) $ with $\mu=(0,1)$ and $\Sigma=$
			$\begin{bmatrix}
			1 & 1\\ \vspace{-.1cm}
			1 & 2  \\  
			\end{bmatrix}$.  For non-scalable coding, we employ the generalized Lloyd algorithm \cite{lbg} to design optimal ECVQ. Note that while Gaussian sources are successively refinable asymptotically, this is not true at finite delays, and hence can benefit from our proposed common information paradigm.
			As shown in Table~\ref{tab:re2}, similar to the scalar case we have significant reduction in total transmit rate compared to non-scalable coding with negligible loss in performance. These results clearly demonstrate the universal effectiveness of the proposed approach to both scalar and vector quantizer design.
			
			We also observed in all the resulting quantizers that even though we did not impose any regularity on the codebook initialization (step 1 of the proposed algorithm), 
			the final quantizers turned to be regular, which is consistent with our lemma.
			\begin{table*}[t]
				\centering
				\scalebox{1}{
					\begin{tabular} {|c|c|c|c|}
						\hline
						&Non scalable total transmit rate & Proposed method total transmit rate &  Total transmit rate reduction \\ 
						&$R_{12}+R_1+R_2= R_t^{NS}$ & $R_{12}+R_1+R_2=R_t^P$ & $(R_t^{NS}-R_t^P)/R_t^{NS} $\\ \hline
						$R_{r_1}=2.3, R_{r_2}=3.5$&$0 + 2.3 + 3.5= 5.8 $& $0.5 + 1.8+ 3 = 5.3$ &  $(5.8-5.3)/5.8= 8$\% \\ \hline
						$R_{r_1}=1.5 , R_{r_2}=3.6$&$0 +1.5+ 3.6 = 5.1 $& $1.1 + 0.4+ 2.5= 4$ &  $(5.1-4)/5.1= 21 $\% \\ \hline
						$R_{r_1}=4.0 , R_{r_2}=5.3$&$0 + 4.0+ 5.3 = 9.3 $& $2.1 + 1.9+ 3.2= 7.2$ &  $(9.3-7.2)/9.3= 22 $\% \\ \hline
						$R_{r_1}=3.1, R_{r_2}=3.8$&$0 + 3.1 + 3.8 = 6.9 $& $1.6 + 1.5+ 2.2 = 5.3$ &  $(6.9-5.3)/6.9= 23$\% \\ \hline
						$R_{r_1}=3.3 , R_{r_2}=4.3$&$0 + 3.3+ 4.3 = 7.6 $& $2.1 + 1.2+ 2.2= 5.5$ &  $(7.6-5.5)/7.6= 27 $\% \\ \hline
						$R_{r_1}=2.6 , R_{r_2}=4.8$&$0 + 2.6+ 4.8 = 7.4 $& $2 + 0.6+ 2.8= 5.4$ &  $(7.4-5.4)/7.4= 27 $\% \\ \hline
						$R_{r_1}=3.7 , R_{r_2}=5.4$&$0 + 3.7+ 5.4 = 9.1 $& $2.5 + 1.2+ 2.9= 6.6$ &  $(9.1-6.6)/9.1= 27 $\% \\ \hline
					\end{tabular}}
					
					\caption{Total transmit rate for non-scalable coding and proposed paradigm operating with negligible loss in distortion for multivariate normal source.}
					\label{tab:re2}
					
				\end{table*}
				\subsection{Results for Multi-Layer Design}
				
				Finally for our last experiment we used the same multivariate Gaussian source for multi layer scenario of $L=3$. In this case encoder generates 5 packets at rates $R_1$, $R_2$, $R_3$, $R_{23}$, and $R_{123}$. Decoder $1$, $2$, and $3$ receive packets at rates $R_{r_1}=R_{1}+R_{123} $, $R_{r_2}=R_{2}+R_{23}+R_{123}$, and $R_{r_3}=R_{3}+R_{23}+R_{123}$  respectively. Total transmit rate for proposed common information paradigm is $R_t = R_1 + R_2 + R_3 + R_{23} + R_{123}$.
				
				As shown in Table~\ref{tab:re3}, similar to the two layer case we have significant reduction in total transmit rate compared to non-scalable coding with negligible loss in performance.
				
				\textbf{Note: } Comparison between two layered ECVQ and the multi-layer case, shows that on average we have more compression efficiency in multi-layer design, with roughly 50 \% reduction in total transmit rate compare to 30 \% reduction for the two layer case.  
				
				\begin{table*}[t]
					\centering
					\scalebox{.9}{
						\begin{tabular} {|c|c|c|c|}
							\hline
							&Non scalable total transmit rate & Proposed method total transmit rate &  Total transmit rate reduction \\ 
							&$R_1+R_2+R_3+R_{23}+R_{123}= R_t^{NS}$ & $R_1+R_2+R_3+R_{23}+R_{123}=R_t^P$ & $(R_t^{NS}-R_t^P)/R_t^{NS} $\\ \hline

						$R_{r_1}=1.2, R_{r_2}=1.5, R_{r_3}=2.0 $&$1.2 + 1.5 + 2.0 + 0 +0 = 4.7 $& $0.2 + 0.2 + 0.7 + 0.3 + 1 = 2.4$ &  $(4.7-2.4)/4.7= 49$\% \\ 
						\hline
						$R_{r_1}=1.2, R_{r_2}=2.0, R_{r_3}=2.3 $&$1.2 + 2.0 + 2.3 + 0 +0 = 5.5 $& $0.2 + 0.4 + 0.7 + 0.6 + 1 = 2.9$ &  $(5.5-2.9)/5.5= 47$\% \\ 
						\hline
						$R_{r_1}=0.8, R_{r_2}=2.6, R_{r_3}=3.3 $&$0.8 + 2.6 + 3.3 + 0 +0 = 6.7 $& $0.2 + 2.0 + 2.7 + 0.0 + 0.6 = 5.5$ &  $(6.7-5.5)/6.7= 22$\% \\ 
						\hline
							$R_{r_1}=2.2, R_{r_2}=2.7, R_{r_3}=2.85 $&$2.2 + 2.7 + 2.85 + 0 +0 = 7.75 $& $1 + 0.55 + 0.7 + 0.95 + 1.2 = 4.4$ &  $(7.75-4.4)/7.75= 43$\% \\ \hline
						$R_{r_1}=3.15, R_{r_2}=3.25, R_{r_3}=3.55 $&$3.15 + 3.25 + 3.55 + 0 +0 = 9.95 $& $0.85 + 0.35 + 0.65 + 0.6 + 2.3 =4.75$ &  $(9.95-4.75)/9.95= 52$\% \\ \hline
						$R_{r_1}=3.3, R_{r_2}=3.4, R_{r_3}=3.6 $&$3.3 + 3.4 + 3.6 + 0 +0 = 10.3 $& $1 + 0.7 + 0.9 + 0.4 + 2.3 = 5.3$ &  $(10.3-5.3)/10.3= 48$\% \\ \hline

						\end{tabular}}
						
						\caption{Total transmit rate for non-scalable coding and proposed paradigm  for 3 layer case, operating with negligible loss in distortion for multivariate normal source.}
						\label{tab:re3}
						
					\end{table*}

	\end{section}


%
\begin{section}{Conclusion}
	\label{sec:concl}
	
	A novel fundamental design technique for a common information based layered coding framework is proposed, wherein jointly designing common and individual layers' vector quantizers overcomes the limitations of conventional scalable coding and non-scalable coding, by providing the flexibility of transmitting common and individual bit-streams for different quality levels. The proposed iterative scalar and vector quantizer design technique optimizes all the quantizers jointly to minimize the overall cost at each iteration. This extracts the common information between different quality levels with negligible performance penalty, and also achieves better operating points in the tradeoff between total transmit rate and distortions at the decoders. Simulation results for the practically important Laplacian sources as well multivariate normal sources, confirm the usefulness of the proposed approach. 
\end{section}

\appendices

\section{}{\label{appendix:a}}
\label{apn:cpart}
To obtain the optimal partition boundaries for the common layer, we take derivative of: 
\[J=  a_1 D_1 + a_2 D_2 + \lambda_1(R_1+R_{12})+\lambda_2(R_2+R_{12})+\lambda_{12}R_{12},\]
with respect to $t_i^{12}$.
Note that, receive rate at decoder $l$ is:
\[R_l+R_{12} = - \sum_{i=1}^{M} \sum_{q=1}^{N_i^l} p_{q,i}^l \log_2 p_{q,i}^l , ~~~\text{for} ~~l=1,2\] and common rate is: 
\[R_{12} = -\sum_{i=1}^M p_i^{12}  \log_2  p_i^{12}.\]
Distortion at individual layer $l$ is:
\[D_l = \sum_{i=1}^{M} \sum_{q=1}^{N_i^l} \int_{t_{q-1,i}^l }^{t_{q,i}^l } (x-x_{q,i}^l )^2 f(x) dx, ~~~\text{for} ~~l=1,2. \]
Also note that $t_{N_i^l,i}^l=t_{i}^{12}$ and $t_{0,i+1}^l=t_{i}^{12}$.
Taking derivative  of $J$ respect to $t_i^{12}$:
\begin{dmath*}
\frac{\partial J}{\partial t_i^{12}}=  
a_1 \frac{\partial D_1}{\partial t_i^{12}} +
a_2 \frac{\partial D_2}{\partial t_i^{12}} +
\lambda_1 \frac{\partial (R_1+R_{12})}{\partial t_i^{12}} +
\lambda_2 \frac{\partial (R_2+R_{12})}{\partial t_i^{12}}+
\lambda_{12} \frac{\partial R_{12}}{\partial t_i^{12}} 
\end{dmath*}
For $l = 1, 2 :$
\begin{dmath*}
\frac{\partial D_l}{\partial t_i^{12}}  =  
\frac{\partial \int_{t_{N_i^l-1,i}^l }^{t_{N_i^l,i}^l } (x-x_{N_i^l,i}^l )^2 f(x) dx}{\partial t_i^{12}} +
\frac{\partial \int_{t_{0,i+1}^l }^{t_{1,i+1}^l } (x-x_{1,i+1}^l )^2 f(x) dx}{\partial t_i^{12}}\\
 = (t_i^{12}-x_{N_i^l,i}^l )^2 f(t_i^{12}) - (t_i^{12}-x_{1,i+1}^l )^2 f(t_i^{12}) \\
 =( ((x_{N_i^l,i}^l)^2-(x_{1,i+1}^l)^2) - 2~ t_i^{12} ({x_{N_i^l,i}^l}-{x_{1,i+1}^l})) f(t_i^{12})
\end{dmath*}
For $l = 1, 2 :$
\begin{dmath*}
-\frac{\partial (R_l+R_{12})}{\partial t_i^{12}}  =  \frac{\partial(p_{N_i^l,i}^l \log_2 p_{N_i^l,i}^l)}{\partial t_i^{12}} +
\frac{\partial(p_{1,i+1}^l \log_2 p_{1,i+1}^l)}{\partial t_i^{12}}\\  =
\frac{d (p_{N_i^l,i}^l \log_2 p_{N_i^l,i}^l)}{d p_{N_i^l,i}^l} * \frac{\partial p_{N_i^l,i}^l}{\partial t_i^{12}} + \frac{d (p_{1,i+1}^l \log_2 p_{1,i+1}^l)}{d p_{1,i+1}^l} * \frac{\partial p_{1,i+1}^l}{\partial t_i^{12}} \\
 = (\log_2 e +\log_2  p_{N_i^l,i}^l) * f(t_i^{12}) - (log_2  e +\log_2  p_{1,i+1}^l) * f(t_i^{12})\\
 = (\log_2  p_{N_i^l,i}^l  - \log_2  p_{1,i+1}^l)* f(t_i^{12})
\end{dmath*}
Similarly: 
\[
-\frac{\partial (R_{12})}{\partial t_i^{12}} = 
(\log_2  p_i^{12} - \log_2  p_{i+1}^{12} )* f(t_i^{12})
\]
By setting $\frac{\partial J}{\partial t_i^{12}}$ to zero and canceling $ f(t_i^{12})$ from all sides  we finally get:

\begin{dmath*}	
	t_i^{12} = \frac{(a_1 (x_{1,i+1}^1)^2+a_2 (x_{1,i+1}^2)^2)-(a_1 (x_{N_i^1,i}^1)^2+a_2 (x_{N_i^2,i}^2)^2)}{2((a_1 {x_{1,i+1}^1}+a_2 {x_{1,i+1}^2})-(a_1 {x_{N_i^1,i}^1}+a_2 {x_{N_i^2,i}^2}))}  - \hspace{-1.2pt}
	\frac{\lambda_1 (\log_2 p_{1,i+1}^1 - \log_2 p_{N_i^1,i}^1) + \lambda_2 (\log_2 p_{1,i+1}^2 - \log_2 p_{N_i^2,i}^2) }{{2((a_1 {x_{1,i+1}^1}+a_2 {x_{1,i+1}^2})-(a_1 {x_{N_i^1,i}^1}+a_2 {x_{N_i^2,i}^2})})}\\
	\hspace{-1.2pt}
	-\frac{ \lambda_{12} (\log_2 p_{i+1}^{12} - \log_2 p_{i}^{12})}{{2((a_1 {x_{1,i+1}^1}+a_2 {x_{1,i+1}^2})-(a_1 {x_{N_i^1,i}^1}+a_2 {x_{N_i^2,i}^2})})}
\end{dmath*}
\hfill  
$\blacksquare$

\section*{Acknowledgment}

The authors would like to thank  NSF funding under the code of  NSF-CCF-1016861 and NSF-CCF-1320599.

\ifCLASSOPTIONcaptionsoff
  \newpage
\fi



%

\bibliographystyle{IEEEtran}
\bibliography{dissertation}

%

\begin{IEEEbiography}[{\includegraphics[width=1in,height=1.25in,clip,keepaspectratio]{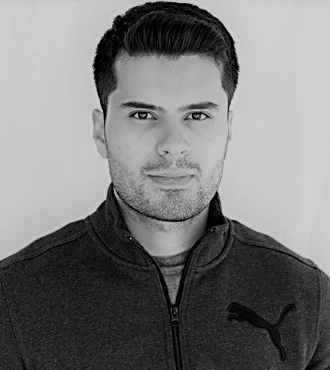}}]{Mehdi Salehifar}
	Mehdi Salehifar (S’17) received the B.Sc degree in electrical and computer engineering from University of Tehran,Tehran,Iran, in 2012 and the M.Sc. and Ph.D.  degrees in electrical and computer engineering from the University of California, Santa Barbara (UCSB), in 2014 and 2017 respectively. He is currently working in LG Electronics as a senior research engineer. His research interests include signal
	processing, general quantization theory, information theory, and video/audio coding. He worked in LG Electronics as a senior video researcher intern as well as at Qualcomm as an Audio researcher. He is a winner of several titles in national and international mathematic competitions.
\end{IEEEbiography}
\begin{IEEEbiography}[{\includegraphics[width=1in,height=1.25in,clip,keepaspectratio]{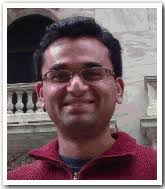}}]{Tejaswi Nanjundaswamy}
	Tejaswi Nanjundaswamy (S’11–M’14) received the B.E degree in electronics and communications engineering from the National Institute of Technology Karnataka, India, in 2004 and the M.S. and Ph.D. degrees in electrical and computer engineering from the University of California, Santa Barbara (UCSB), in 2009 and 2013, respectively. He is currently a post-doctoral researcher at Signal Compression Lab in UCSB, where he focuses on audio/video compression, processing and related technologies. He worked at Ittiam Systems, Bangalore, India from 2004 to 2008 as Senior Engineer on audio codecs and effects development. He also interned in the Multimedia Codecs division of Texas Instruments (TI), India in 2003.Dr. Nanjundaswamy is an associate member of the Audio Engineering Society (AES). He won the Student Technical Paper Award at the AES 129th Convention.
\end{IEEEbiography}
\begin{IEEEbiography}[{\includegraphics[width=1in,height=1.25in,clip,keepaspectratio]{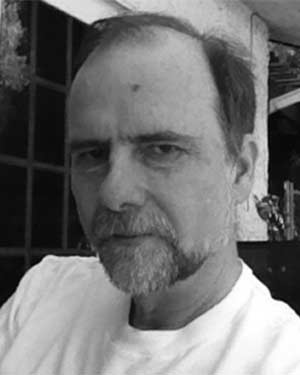}}]{Kenneth Rose}
	Kenneth Rose (S’85–M’91–SM’01–F’03) received the Ph.D. degree from the California Institute of Technology, Pasadena, in 1991.He then joined the Department of Electrical and Computer Engineering, University of California at Santa Barbara, where he is currently a Professor. His main research activities are in the areas of information theory and signal processing, and include rate-distortion theory, source and source-channel coding, audio-video coding and networking, pattern recognition, and non-convex optimization. He is interested in the relations between information theory, estimation theory, and statistical physics, and their potential impact on fundamental and practical problems in diverse disciplines.Prof. Rose was corecipient of the 1990 William R. Bennett Prize Paper Award of the IEEE Communications Society, as well as the 2004 and 2007 IEEE Signal Processing Society Best Paper Awards.
\end{IEEEbiography}




\end{document}